\documentclass[11pt]{article}

\usepackage[margin=1in]{geometry}
\usepackage{CJK, hyperref, amsmath, amsthm, amssymb, authblk}

\newtheorem{theorem}{Theorem}

\theoremstyle{definition}
\newtheorem{definition}{Definition}

\theoremstyle{remark}
\newtheorem*{remark}{Remark}

\DeclareMathOperator{\poly}{poly}
\DeclareMathOperator{\tr}{tr}

\begin{document}

\title{Absence of temporal order in states with spatial correlation decay}

\begin{CJK*}{UTF8}{}

\CJKfamily{gbsn}

\author{Yichen Huang (黄溢辰)\thanks{yichuang@mit.edu}\\
Center for Theoretical Physics, Massachusetts Institute of Technology\\
Cambridge, Massachusetts 02139, USA\\
}

\maketitle

\end{CJK*}

\begin{abstract}

In quantum lattice systems, we prove that any stationary state with power-law (or even exponential) decay of spatial correlations has vanishing macroscopic temporal order in the thermodynamic limit. Assuming translational invariance, we obtain a similar bound on the temporal order between local operators at late times. Our proofs do not require any locality of the Hamiltonian. Applications in quantum time crystals are briefly discussed.

\end{abstract}

\section{Introduction and preliminaries}

We study temporal correlations, which are of particular interest in the context of quantum dynamics \cite{SS15, MSS16} and time crystals \cite{SZ17, EMNY19, KMS19}.

Without loss of generality, we work with a hypercubic lattice in $D=O(1)$ spatial dimensions such that each lattice site corresponds to a point in $\mathbb Z^D$. (It is easy to see that the same results hold for other lattices.) Suppose there is a spin at every lattice site. The system size, defined as the total number of sites or spins in the lattice, is denoted by $N=n^D$, where $n$ (assumed to be an integer) is the length of the lattice. We always consider the thermodynamic limit $n\to+\infty$.

The lattice induces a metric that allows us to define correlation decay. Let
\begin{equation}
\mathbf i=(i_1,i_2,\ldots,i_D)\in\mathbb Z^D,\quad\mathbf i'=(i'_1,i'_2,\ldots,i'_D)\in\mathbb Z^D
\end{equation}
be two sites. Their distance is given by
\begin{equation}
|\mathbf i-\mathbf i'|:=\sum_{l=1}^D|i_l-i'_l|\in\mathbb Z.
\end{equation}

Throughout this paper, asymptotic notations are used extensively. Let $f,g:\mathbb R^+\to\mathbb R^+$ be two positive functions. One writes $f(x)=O(g(x))$ if and only if there exist positive numbers $M,x_0$ such that $f(x)\le Mg(x)$ for all $x>x_0$; $f(x)=\Omega(g(x))$ if and only if there exist positive numbers $M,x_0$ such that $f(x)\ge Mg(x)$ for all $x>x_0$; $f(x)=\Theta(g(x))$ if and only if there exist positive numbers $M_1,M_2,x_0$ such that $M_1g(x)\le f(x)\le M_2g(x)$ for all $x>x_0$.

\begin{definition} [two-point correlation decay]
Suppose that a density matrix $\rho$ satisfies
\begin{equation}
|\tr(\rho L_{\mathbf i}L_{\mathbf i'})-\tr(\rho L_{\mathbf i})\tr(\rho L'_{\mathbf i'})|\le\|L_{\mathbf i}\|\|L_{\mathbf i'}\|f(|\mathbf i-\mathbf i'|)
\end{equation}
for any local operators $L_{\mathbf i},L'_{\mathbf i'}$ at arbitrary sites $\mathbf i,\mathbf i'$. The state $\rho$ has power-law or exponential decay of spatial correlations if $f(r)=1/\poly r$ or $f(r)=e^{-\Omega(r)}$, respectively, in the limit $r\to+\infty$.
\end{definition}

Let $A_{\mathbf i},B_{\mathbf i}$ with $\|A_{\mathbf i}\|,\|B_{\mathbf i}\|=O(1)$ be (not necessarily Hermitian) local operators supported in a small neighborhood of the site $\mathbf i$. Define the ``macroscopic'' operators
\begin{equation}
A=\frac{1}{N}\sum_{\mathbf i}A_{\mathbf i},\quad B=\frac{1}{N}\sum_{\mathbf i}B_{\mathbf i}.
\end{equation}
Note that $A,B$ do not have to be translationally invariant.

In this paper, we study the (connected) temporal correlators between the macroscopic operators $A$ and $B$ and between the local operators $A_{\mathbf i}$ and $B_{\mathbf i'}$ for stationary states with power-law (or even exponential) decay of spatial correlations. (A density matrix $\rho$ is stationary if it does not evolve under the Hamiltonian, i.e., $[\rho,H]=0$.) In the thermodynamic limit, we prove that the former vanishes at any time. Furthermore, the latter vanishes at late times assuming translational invariance and the non-degeneracy of the spectrum of $H$. Our proofs do not use the Lieb-Robinson bound \cite{LR72, NS06, HK06}, nor any other notion of the locality of $H$. Therefore, our results apply to Hamiltonians with arbitrary long-range interactions. 

Many physical states have spatial correlation decay. For example, ground states of gapped local Hamiltonians have a finite correlation length \cite{NS06, HK06}, while critical states usually have power-law correlation decay. The thermal state of a local Hamiltonian is expected to have a finite correlation length above the critical temperature. One can prove that one-dimensional quantum systems always have a finite correlation length at any temperature \cite{Ara69}, while in two and higher dimensions we have exponential decay of correlations at sufficiently high temperature \cite{KGK+14}.

\section{Macroscopic operators}

We are ready to present the main results of this paper.

\begin{theorem} \label{main}
Let $\rho$ be a stationary state with correlation decay $f(r)=O(r^{-\alpha})$. At any time $t\in\mathbb R$,
\begin{equation}
|\tr(\rho A(t)B)-\tr(\rho A)\tr(\rho B)|=
\begin{cases}
O(1/N) & \alpha>D\\
O(\log N)/N & \alpha=D\\
O(n^{-\alpha}) & \alpha<D
\end{cases},
\end{equation}
where $A(t):=e^{iHt}Ae^{-iHt}$ is the time evolution of $A$ in the Heisenberg picture, and $N=n^D$ is the system size. Note that $A,B$ do not have to be translationally invariant, and $H$ does not have to be a local Hamiltonian.
\end{theorem}

\begin{proof}
We may assume $\tr(\rho A_{\mathbf i})=\tr(\rho B_{\mathbf i})=0$ for any $\mathbf i$. This is without loss of generality because one can simply use the transform
\begin{equation}
A_{\mathbf i}\rightarrow A_{\mathbf i}-\tr(\rho A_{\mathbf i})I,\quad B_{\mathbf i}\rightarrow B_{\mathbf i}-\tr(\rho B_{\mathbf i})I
\end{equation}
if necessary. Let $\{|1\rangle,|2\rangle,\ldots\}$ be a complete set of eigenstates of $H$ with corresponding energies $E_1\le E_2\le\cdots$, and $X_{jk}=\langle j|X|k\rangle$ be the matrix element of an operator in the energy eigenbasis. Writing out the matrix elements,
\begin{multline}
|\tr(\rho A(t)B)|=\left|\sum_{j,k}\rho_{jj}A_{jk}B_{kj}e^{i(E_j-E_k)t}\right|\le\sum_{j,k}\rho_{jj}|A_{jk}||B_{kj}|\le\sqrt{\sum_{j,k}\rho_{jj}|A_{jk}|^2\times\sum_{j,k}\rho_{jj}|B_{kj}|^2}\\
=\sqrt{\sum_j\rho_{jj}(AA^\dag)_{jj}\times\sum_j\rho_{jj}(B^\dag B)_{jj}}=\sqrt{\tr(\rho AA^\dag)\tr(\rho B^\dag B)}.
\end{multline}
Consider the first factor:
\begin{equation}
\tr(\rho AA^\dag)\le\frac{1}{N^2}\sum_{\mathbf i,\mathbf i'}|\tr(\rho A_{\mathbf i}A_{\mathbf i'}^\dag)|=\frac{O(1)}{N^2}\sum_{\mathbf i,\mathbf i'}f(|\mathbf i-\mathbf i'|)\le\frac{O(1)}{N}\sum_{|\mathbf i|=O(n)}f(|\mathbf i|).
\end{equation}
The series can be estimated from an integral
\begin{equation}
\frac{O(1)}{N}\sum_{|\mathbf i|=O(n)}|\mathbf i|^{-\alpha}=\frac{O(1)}{N}\int_{r=1}^{O(n)}r^{-\alpha}\times r^{D-1}\,\mathrm dr=\frac{O(1)}{N}\times
\begin{cases}
O(1) & \alpha>D\\
O(\log n) & \alpha=D\\
O(n^{D-\alpha}) & \alpha<D
\end{cases}.
\end{equation}
We complete the proof by noting that $\tr(\rho B^\dag B)$ can be upper bounded similarly.
\end{proof}

\begin{remark}
The bound $O(1/N)$ for $\alpha>D$ in Theorem \ref{main} is tight for any $D=O(1)$. For example, let $H=\sum_{\mathbf i}H_{\mathbf i}$ be a translationally invariant local Hamiltonian, where each term is traceless $\tr H_{\mathbf i}=0$ and has operator norm $\|H_{\mathbf i}\|=\Theta(1)$. Let $A=B=H/N$ and $\rho=I/\tr I$ be the infinite temperature state. Then,
\begin{equation}
\tr(\rho A(t)B)=\tr(H^2)/(N^2\tr I)=\Theta(1/N).
\end{equation}
The last step is a well-known fact, and can be proved by expanding $H$ in the Pauli operator basis and then counting the number of terms that do not vanish upon taking the trace in the expansion of $H^2$; see, e.g., Ref. \cite{HBZ19}.
\end{remark}

\section{Local operators}

We now consider the temporal correlation of local operators at large time $t$.

\begin{theorem} \label{local}
Suppose that the Hamiltonian $H$ is translationally invariant and its spectrum is simple (i.e., non-degenerate). Let $\rho$ be a stationary state with correlation decay $f(r)=O(r^{-\alpha})$. Then,
\begin{equation}
\left|\lim_{\tau\to+\infty}\frac{1}{\tau}\int_0^\tau\tr(\rho A_{\mathbf i}(t)B_{\mathbf i'})\,\mathrm dt-\tr(\rho A_{\mathbf i})\tr(\rho B_{\mathbf i'})\right|=
\begin{cases}
O(1/N) & \alpha>D\\
O(\log N)/N & \alpha=D\\
O(n^{-\alpha}) & \alpha<D
\end{cases}
\end{equation}
for any local operators $A_{\mathbf i},B_{\mathbf i'}$ at arbitrary sites $\mathbf i,\mathbf i'$. Note that $H$ does not have to be a local Hamiltonian.
\end{theorem}

\begin{proof}
The translational invariance of $H$ implies that all eigenstates are translationally invariant. Since the statement of the theorem does not involve $A,B$, we can simply define them to be translationally invariant! This is achieved by taking the sum of the lattice-translated copies of $A_{\mathbf i}$ and $B_{\mathbf i'}$, respectively, and dividing by $N$ as before. Thus, $(A_{\mathbf i})_{jj}=A_{jj}$ and $(B_{\mathbf i'})_{jj}=B_{jj}$. We still use the convention $\tr(\rho A_{\mathbf i})=\tr(\rho B_{\mathbf i'})=0$. Writing out the matrix elements,
\begin{multline}
\left|\lim_{\tau\to+\infty}\frac{1}{\tau}\int_0^\tau\tr(\rho A_{\mathbf i}(t)B_{\mathbf i'})\,\mathrm dt\right|=\left|\lim_{\tau\to+\infty}\frac{1}{\tau}\int_0^\tau\sum_{j,k}\rho_{jj}(A_{\mathbf i})_{jk}(B_{\mathbf i'})_{kj}e^{i(E_j-E_k)t}\,\mathrm dt\right|\\
=\left|\sum_{j,k}\rho_{jj}(A_{\mathbf i})_{jk}(B_{\mathbf i'})_{kj}\delta_{E_j,E_k}\right|=\left|\sum_j\rho_{jj}(A_{\mathbf i})_{jj}(B_{\mathbf i'})_{jj}\right|=\left|\sum_j\rho_{jj}A_{jj}B_{jj}\right|\le\sum_{j,k}\rho_{jj}|A_{jk}||B_{kj}|.
\end{multline}
where $\delta$ is the Kronecker delta, and we used the assumption of a simple spectrum. The remaining steps follow those in the proof of Theorem \ref{main}.
\end{proof}

\begin{remark}
The bound $O(1/N)$ for $\alpha>D$ in Theorem \ref{local} is also tight for any $D=O(1)$. This follows from a very similar argument as that in the previous remark on the tightness of Theorem \ref{main}. 
\end{remark}

We compare Theorem \ref{local} with a recent result in the literature.

\begin{theorem} [\cite{ARG19}] \label{thm:ARG}
Suppose that $H$ is a translationally invariant local Hamiltonian and its spectrum is simple. Let $\rho$ be a stationary state with correlation decay $f(r)=O(e^{-r/\xi})$. Then,
\begin{equation} \label{eq:ARG}
\left|\lim_{\tau\to+\infty}\frac{1}{\tau}\int_0^\tau\tr(\rho A_{\mathbf i}(t)B_{\mathbf i'})\,\mathrm dt-\tr(\rho A_{\mathbf i})\tr(\rho B_{\mathbf i'})\right|=O\left(\xi^{\frac{D}{D+1}}N^{-\frac{1}{D+1}}\log N\right).
\end{equation}
for any local operators $A_{\mathbf i},B_{\mathbf i'}$ at arbitrary sites $\mathbf i,\mathbf i'$. 
\end{theorem}

\begin{remark}
We observe that by revising the proof in Ref. \cite{ARG19}, the right-hand side of Eq. (\ref{eq:ARG}) can be improved to
\begin{equation} \label{eq:H}
O\left(\xi^{\frac{2D}{D+1}}N^{-\frac{2}{D+1}}\log^2N\right),
\end{equation}
which remains to be weaker than the bound $O(1/N)$ for $\alpha>D$ in Theorem \ref{local}. We show how to obtain (\ref{eq:H}) in the appendix.
\end{remark}

While both Theorems \ref{local}, \ref{thm:ARG} establish that
\begin{equation}
\lim_{\tau\to+\infty}\frac{1}{\tau}\int_0^\tau\tr(\rho A_{\mathbf i}(t)B_{\mathbf i'})\,\mathrm dt\to\tr(\rho A_{\mathbf i})\tr(\rho B_{\mathbf i'})
\end{equation}
in the thermodynamic limit $n\to+\infty$, it should be clear that Theorem \ref{local} is technically stronger.

\section{Remarks on quantum time crystal}

In a remarkable paper, Watanabe and Oshikawa \cite{WO15} related the temporal correlator between macroscopic operators to the concept of quantum time crystals. In this sense, Theorem \ref{main} can be viewed as a proof of the absence of quantum time crystals for states with spatial correlation decay. This proof complements other proofs \cite{WO15, KMS19} in the literature. At least, our proof does not use the Lieb-Robinson bound and thus applies to time-independent Hamiltonians with arbitrary long-range interactions.

The only severe (and perhaps unfavorable) assumption of our approach is the decay of spatial correlations. However, this assumption is in some sense necessary because it is possible to construct states with long-range correlations in systems with long-range interactions such that a quantum time crystal is observed \cite{KK19}.

Our bounds in Theorem \ref{main} on the temporal correlation between macroscopic operators have the desirable property that it is time-independent. In contrast, the previous bounds \cite{WO15,KMS19} grow with time and cannot rule out the following possibilities:
\begin{itemize}
\item The temporal correlator oscillates with a period that grows with the system size.
\item The temporal correlator starts to oscillate (with a constant period) only after a transition time that grows with the system size. 
\end{itemize}

\emph{Note added.}---Very recently, we became aware of a conceptually related but technically completely different paper \cite{WOK19}, which also proved upper bounds on the temporal correlations between macroscopic operators. The settings in this reference are not the same as ours and hence the results there are not directly comparable to the ones in the present paper except at infinite temperature. In this case, Theorem \ref{main} is technically stronger than Eq. (38) in Ref. \cite{WOK19}, which gives a time-dependent upper bound $O(1+t^{D+1})/N$.

\section*{Acknowledgments}

This work was supported by NSF PHY-1818914.

\appendix

\section{Proof of (\ref{eq:H})}

We show how to obtain (\ref{eq:H}) by revising the proof in Ref. \cite{ARG19}. This appendix is not self-contained: We assume that the reader is already familiar with the notations in this reference. We follow the proof in Ref. \cite{ARG19} up to Eq. (A15), and then do the following.

Define $A_{kk}-\tr(\rho A)=\Delta_{k,A}$ and $B_{kk}-\tr(\rho B)=\Delta_{k,B}$ so that
\begin{multline}
\sum_kA_{kk}B_{kk}\rho_{kk}=\tr(\rho A)\tr(\rho B)+\sum_k\rho_{kk}(\tr(\rho A)\Delta_{k,B}+\tr(\rho B)\Delta_{k,A}+\Delta_{k,A}\Delta_{k,B})\\
=\tr(\rho A)\tr(\rho B)+\sum_k\rho_{kk}\Delta_{k,A}\Delta_{k,B}
\end{multline}
Let $\Delta:=KN^{-\frac{1}{D+1}}\log N$, and split the error term as
\begin{equation} \label{A18}
\left|\sum_k\rho_{kk}\Delta_{k,A}\Delta_{k,B}\right|\le\left|\sum_{k\in S}\rho_{kk}\Delta_{k,A}\Delta_{k,B}\right|+\left|\sum_{k\not\in S}\rho_{kk}\Delta_{k,A}\Delta_{k,B}\right|,
\end{equation}
where $S:=\{k:|\Delta_{k,A}|,|\Delta_{k,B}|\le\Delta\}$. Note that the absolute value of first term is smaller than $\Delta^2$ by definition, and the second term can be bounded as
\begin{multline}
\left|\sum_{k\not\in S}\rho_{kk}\Delta_{k,A}\Delta_{k,B}\right|\le\max_{k'}|\Delta_{k',A}\Delta_{k',B}|\sum_{k\not\in S}\rho_{kk}\le\max_{k'}|\Delta_{k',A}\Delta_{k',B}|\left(\sum_{|\Delta_{k,A}|\ge\Delta}\rho_{kk}+\sum_{|\Delta_{k,B}|\ge\Delta}\rho_{kk}\right)\\
\le2\max_{k'}|\Delta_{k',A}\Delta_{k',B}|e^{-c\Delta N^{\frac{1}{D+1}}\xi^{-\frac{D}{D+1}}}\le O(1)\|A\|\|B\|N^{-cK\xi^{-\frac{D}{D+1}}},
\end{multline}
where we used Lemma 1 in Ref. \cite{ARG19}. We choose $K$ such that $cK\xi^{-\frac{D}{D+1}}=2/(D+1)$. Then, the first term dominates on the right-hand side of (\ref{A18}). Hence,
\begin{equation}
\left|\sum_kA_{kk}B_{kk}\rho_{kk}-\tr(\rho A)\tr(\rho B)\right|=\left|\sum_k\rho_{kk}\Delta_{k,A}\Delta_{k,B}\right|=O(\Delta^2).
\end{equation}
This concludes the proof.

\bibliographystyle{abbrv}
\bibliography{time}

\begin{thebibliography}{10}

\bibitem{ARG19}
A.~M. Alhambra, J.~Riddell, and L.~P. Garc\'{\i}a-Pintos.
\newblock Time evolution of correlation functions in quantum many-body systems.
\newblock arXiv:1906.11280, 2019.

\bibitem{Ara69}
H.~Araki.
\newblock Gibbs states of a one dimensional quantum lattice.
\newblock {\em Communications in Mathematical Physics}, 14(2):120--157, 1969.

\bibitem{EMNY19}
D.~V. Else, C.~Monroe, C.~Nayak, and N.~Y. Yao.
\newblock Discrete time crystals.
\newblock arXiv:1905.13232, 2019.

\bibitem{HK06}
M.~B. Hastings and T.~Koma.
\newblock Spectral gap and exponential decay of correlations.
\newblock {\em Communications in Mathematical Physics}, 265(3):781--804, 2006.

\bibitem{HBZ19}
Y.~Huang, F.~G. S.~L. Brand\~ao, and Y.-L. Zhang.
\newblock Finite-size scaling of out-of-time-ordered correlators at late times.
\newblock {\em Physical Review Letters}, 123(1):010601, 2019.

\bibitem{KMS19}
V.~Khemani, R.~Moessner, and S.~L. Sondhi.
\newblock A brief history of time crystals.
\newblock arXiv:1910.10745, 2019.

\bibitem{KGK+14}
M.~Kliesch, C.~Gogolin, M.~J. Kastoryano, A.~Riera, and J.~Eisert.
\newblock Locality of temperature.
\newblock {\em Physical Review X}, 4(3):031019, 2014.

\bibitem{KK19}
V.~K. Kozin and O.~Kyriienko.
\newblock Quantum time crystals from {H}amiltonians with long-range
  interactions.
\newblock {\em Physical Review Letters}, 123(21):210602, 2019.

\bibitem{LR72}
E.~H. Lieb and D.~W. Robinson.
\newblock The finite group velocity of quantum spin systems.
\newblock {\em Communications in Mathematical Physics}, 28(3):251--257, 1972.

\bibitem{MSS16}
J.~Maldacena, S.~H. Shenker, and D.~Stanford.
\newblock A bound on chaos.
\newblock {\em Journal of High Energy Physics}, 2016(8):106, 2016.

\bibitem{NS06}
B.~Nachtergaele and R.~Sims.
\newblock Lieb-robinson bounds and the exponential clustering theorem.
\newblock {\em Communications in Mathematical Physics}, 265(1):119--130, 2006.

\bibitem{SZ17}
K.~Sacha and J.~Zakrzewski.
\newblock Time crystals: a review.
\newblock {\em Reports on Progress in Physics}, 81(1):016401, 2017.

\bibitem{SS15}
S.~H. Shenker and D.~Stanford.
\newblock Stringy effects in scrambling.
\newblock {\em Journal of High Energy Physics}, 2015(5):132, 2015.

\bibitem{WO15}
H.~Watanabe and M.~Oshikawa.
\newblock Absence of quantum time crystals.
\newblock {\em Physical Review Letters}, 114(25):251603, 2015.

\bibitem{WOK19}
H.~Watanabe, M.~Oshikawa, and T.~Koma.
\newblock Proof of the absence of long-range temporal orders in {G}ibbs states.
\newblock arXiv:1911.12939, 2019.

\end{thebibliography}

\end{document}